\documentclass[journal]{IEEEtran}

\usepackage{amssymb}
\usepackage{enumitem}
\usepackage{amsmath}
\usepackage{amsthm}
\usepackage{amssymb}
\usepackage{graphicx}
\usepackage{epstopdf}
\usepackage{latexsym}
\usepackage{cite}
\usepackage[dvips]{epsfig}
\usepackage{mathrsfs}
\usepackage{dsfont}
\usepackage{fix2col}
\usepackage{algorithm}
\usepackage{algpseudocode}
\usepackage{color,subcaption}
\usepackage{fancyhdr}

\newcommand{\be}{\begin{eqnarray}}

\newcommand{\ee}{\end{eqnarray}}
\newcommand{\ben}{\begin{eqnarray*}}
	\newcommand{\een}{\end{eqnarray*}}
\newcommand{\bfl}{\begin{flalign*}}
	\newcommand{\efl}{\end{flalign*}}
\newcommand{\dref}[1]{(\ref{#1})}

\newcommand{\calJ}{{\mathcal J}}
\newcommand{\calI}{{\mathcal I}}

\newcommand{\calK}{{\mathcal K}}

\newtheorem{theorem}{Theorem}
\newtheorem{definition}{Definition}
\newtheorem{lemma}{Lemma}
\newtheorem{corollary}{Corollary}

\newtheorem{example}{Example}

\graphicspath{{Figures/}}

\ifCLASSINFOpdf

\else

\fi

\begin{document}
\title{An Improved Bound for Minimizing the Total Weighted Completion Time of Coflows in Datacenters
	}

\author{Mehrnoosh~Shafiee,
        and~Javad~Ghaderi\\
        Electrical Engineering Department\\
        Columbia University
\thanks{An earlier version of this paper appears as a brief announcement in SPAA 2017 Conference~\cite{ourspaawork}. Authors' emails: s.mehrnoosh@columbia.edu,  and jghaderi@ee.columbia.edu. }}

\maketitle

\begin{abstract}
	In data-parallel computing frameworks, intermediate parallel data is often produced at various stages which needs to be transferred among servers in the datacenter network (e.g. the shuffle phase in MapReduce). A stage often cannot start or be completed unless all the required data pieces from the preceding stage are received. \emph{Coflow} is a recently proposed networking abstraction to capture such communication patterns. We consider the problem of efficiently scheduling coflows with release dates in a shared datacenter network so as to minimize the total weighted completion time of coflows.
	Several heuristics have been proposed recently to address this problem, as well as a few polynomial-time approximation algorithms with provable performance guarantees. Our main result in this paper is a polynomial-time deterministic algorithm that improves the prior known results. Specifically, we propose a deterministic algorithm with approximation ratio of $5$, which improves the prior best known ratio of $12$. For the special case when all coflows are released at time zero, our deterministic algorithm obtains approximation ratio of $4$ which improves the prior best known ratio of $8$. The key ingredient of our approach is an improved linear program formulation for sorting the coflows followed by a simple list scheduling policy. Extensive simulation results, using both synthetic and real traffic traces, are presented that verify the performance of our algorithm and show improvement over the prior approaches.
\end{abstract}


\IEEEpeerreviewmaketitle

\section{Introduction}
\label{introduction}
Many data-parallel computing applications (e.g. MapReduce~\cite{dean2008mapreduce}, Hadoop~\cite{shvachko2010hadoop,borthakur2007hadoop}, Dryad~\cite{isard2007dryad}, etc.) consist of multiple computation and communication stages or have
machines grouped by functionality. While computation involves local operations in servers, communication takes place at the level of machine groups and involves transfer of many pieces of intermediate data across groups of machines for further processing. In such applications, the collective effect of all the flows between the two machine groups is more important than that of any of the individual flows.
A computation stage often cannot start unless all the required data pieces from the preceding stage are received, or the application latency is determined by the transfer of the last flow between the groups~\cite{chowdhury2014efficient,dogar2014decentralized}. For example, consider a MapReduce application: Each mapper performs local computations and writes intermediate data to
the disk, then each reducer pulls intermediate data from different mappers, merges them, and computes its output. The job will not finish until its last reducer is completed. Consequently, the job completion time depends on the time that the last flow of the communication phase (called shuffle) is finished. Such intermediate communication stages in a data-parallel application can account for more than $50\%$ of the job completion time~\cite{chowdhury2011managing}, and hence can have a significant impact on application performance.
Optimizing flow-level performance metrics (e.g. the average flow completion time) have been extensively studied before from both networking systems and theoretical perspective (see, e.g.,~\cite{alizadeh2014conga, he2015presto, shafiee2016simple} and references there.), however, these metrics ignore the dependence among the flows of an application which is required for the application-level performance.

Recently Chowdhury and Stoica~\cite{chowdhury2012coflow} have introduced the \textit{coflow} abstraction to capture these communication patters. \textit{A coflow is defined as a collection of parallel flows whose completion time is determined by the completion time of the last flow in the collection}. Coflows can represent most communication patterns between successive computation stages of data-parallel applications~\cite{chowdhury2014efficient}. Clearly the traditional flow communication is still a coflow with a single flow.
Jobs from one or more data-parallel applications create multiple coflows in a shared datacenter network. These coflows could vary widely in terms of the total size of the parallel flows, the number of the parallel flows, and the size of the individual flows in the coflows (e.g., see the analysis of production traces in~\cite{chowdhury2014efficient}). Classical flow/job scheduling algorithms do not perform well in this environment~\cite{chowdhury2014efficient} because each coflow consists of multiple flows-- whose completion time is dominated by its slowest flow-- and further, the progress of each flow depends on its assigned rate at \textit{both} its source and its destination. This coupling of rate assignments between the flows in a coflow and across the source-destination pairs in the network is what makes the coflow scheduling problem considerably harder than the classical flow/job scheduling problems.

In this paper, we study the coflow scheduling problem, namely, the algorithmic task of determining when to start serving each flow and at what rate, in order to minimize the weighted sum of completion times of coflows in the system. In the case of equal weights, this is equivalent to minimizing the average completion time of coflows.
\subsection{Related Work}
Several scheduling heuristics have been already proposed in the literature for scheduling coflows, e.g.~\cite{dogar2014decentralized, chowdhury2014efficient,zhao2015rapier,chowdhury2015efficient}. A FIFO-based solution was proposed in~\cite{dogar2014decentralized} which also uses multiplexing of coflows to avoid starvation of small flows which are blocked by large head-of-line flows.
A Smallest-Effective-Bottleneck-First heuristic was introduced in Varys~\cite{chowdhury2014efficient}: it sorts the coflows in an ascending order in a list based on their maximum loads on the servers, and then assigns rates to the flows of the first coflow in the list such that all its flows finish at the same time. The remaining capacity is distributed among the rest of the coflows in the list in a similar fashion to avoid under-utilization of the network. Similar heuristics without prior knowledge of coflows were introduced in Aalo~\cite{chowdhury2015efficient}. A joint scheduling and routing of coflows in datacenter networks was introduced in~\cite{zhao2015rapier} where similar heuristics based on a Minimum-Remaining-Time-First policy are developed. In~\cite{chenoptimizing}, an algorithm is proposed to provide max-min fairness among the colfows.

Here, we would like to highlight three papers~\cite{qiu2015minimizing,khuller2016brief,oursigmetricswork}
that are more relevant to our work. These papers consider the problem of minimizing the total weighted completion time of coflows with release dates (i.e., coflows arrive over time.) and provide algorithms with provable guarantees. This problem is shown to be NP-complete through its connection with the concurrent open shop problem~\cite{chowdhury2014efficient,qiu2015minimizing}, and then approximation algorithms are proposed which run in polynomial time and return a solution whose value is guaranteed to be within a constant fraction of the optimal (a.k.a., \textit{approximation ratio}). These papers rely on linear programming relaxation techniques from combinatorial scheduling literature (see, e.g.,~\cite{jurcik2009open,kim2005data,hall1996scheduling}). In~\cite{qiu2015minimizing}, the authors utilize an interval-indexed linear program formulation which helps partitioning the coflows into disjoint groups. All coflows that fall into one partition are then viewed as a single coflow, where a polynomial-time algorithm is used to optimize its completion time.
Authors in~\cite{khuller2016brief} have recently constructed an instance of the concurrent open shop problem (see~\cite{mastrolilli2010minimizing} for the problem definition) from the original coflow scheduling problem. Then applying the well-known approximation algorithms for the concurrent open shop problem to the constructed instance, an ordering of coflows is obtained which is then used in a similar fashion as in~\cite{qiu2015minimizing} to obtain an approximation algorithm. The deterministic algorithm in~\cite{khuller2016brief} has better approximation ratios compared to~\cite{qiu2015minimizing}, for both cases of with and without release dates.
In~\cite{oursigmetricswork}, a linear program approach based on ordering variables is utilized to develop two algorithms, one deterministic and the other randomized. The deterministic algorithm gives the same bounds as in~\cite{khuller2016brief}, while the randomized algorithm has better performance approximation ratios compared to~\cite{qiu2015minimizing,khuller2016brief}, for both cases of with and without release dates.
\subsection{Contribution}
\label{mainresults}
In this paper, we consider the problem of minimizing the total weighted coflow completion time. Our main contributions can be summarized as follows.

\noindent $\bullet$ \textbf{Scheduling Algorithm.} We use a Linear Program (LP) approach based on ordering variables followed by a simple list scheduling policy to develop a deterministic algorithm. Our approach improves the prior algorithms in both cases of with and without release dates. Table~\ref{comparison} summarizes our results in comparison with the prior best-known performance bounds. Performance of a deterministic (randomized) algorithm is defined based on approximation ratio, i.e., the ratio between the (expected) weighted sum of coflow completion times obtained by the algorithm and the optimal value. When coflows have release dates (which is often the case in practice as coflows are generated at different times), our deterministic algorithm improves the approximation ratio of $12$~\cite{khuller2016brief,oursigmetricswork} to $5$, which is also better than the best known randomized algorithm proposed in~\cite{oursigmetricswork} with approximation ratio of $3e$ ($\approx 8.16$). When all coflows have release times equal to zero, our deterministic algorithm has approximation ratio of $4$ while the best prior known result is $8$~\cite{khuller2016brief,oursigmetricswork} for deterministic and $2e$ $(\approx 5.436)$~\cite{oursigmetricswork} for randomized algorithms.

\noindent $\bullet$ \textbf{Empirical Evaluations.} We evaluate the performance of our algorithm, compared to the prior approaches, using both syntectic traffic as well as real traffic based on a Hive/MapReduce trace from a large production cluster at Facebook. Both synthetic and empirical evaluations show that our deterministic algorithm indeed outperforms the prior approaches. For instance, for the Facebook trace with general release dates, our algorithm outperforms Varys~\cite{chowdhury2014efficient}, deterministic algorithm proposed in~\cite{qiu2015minimizing}, and deterministic algorithm proposed in~\cite{oursigmetricswork} by $24\%$, $40\%$, and $19\%$, respectively.
\begin{table}
	\caption{Performance guarantees (Approximation ratios)}
	\centering
	\begin{tabular}{c c c} % centered columns (4 columns)
		\hline %inserts double horizontal lines
		Case & Best known & This paper \\ 
& deterministic~~~randomized& deterministic\\[0.5ex] % inserts table
		%heading
		\hline % inserts single horizontal line
\hline
		Without release dates & $8$~\cite{khuller2016brief,oursigmetricswork}~~~ $2e$~\cite{oursigmetricswork} & $4$ \\ 
		With release dates & $12$~\cite{khuller2016brief,oursigmetricswork}~~~ $3e$~\cite{oursigmetricswork}& $5$ \\
		\hline %inserts single line
	\end{tabular}
	\label{comparison}
\end{table}

\section{System Model and Problem Formulation}
\label{ProbState}
\subsection*{Datecenter Network}
Similar to~\cite{chowdhury2014efficient,qiu2015minimizing}, we abstract out the datacenter network as one giant $N \times N$ non-blocking switch, with $N$ input links connected to $N$ source servers and $N$ output links connected to $N$ destination servers. Thus the network can be viewed as a bipartite graph with source nodes denoted by set $\calI$ on one side and destination nodes denoted by set $\calJ$ on the other side. Moreover, there are capacity constraints on the input and output links. For simplicity, we assume all links have equal capacity (as in~\cite{qiu2015minimizing}); nevertheless, our method can be easily extended to the general case where the links have unequal capacities. Without loss of generality, we assume that all the link capacities are normalized to one.
\subsection*{Scheduling Constraints}
We allow a general class of scheduling algorithms where the rate allocation can be performed continuously over time, i.e., for each flow, fractional data units can be transferred from its input link to its corresponding output link over time as long as link capacity constraints are respected. In the special case that the rate allocation is restricted to data units (packets), each source node can send at most one packet in every time unit (time slot) and each destination node can receive at most one packet in every time slot, and the feasible schedule has to form a matching of the switch's bipartite graph. In this case, our model reduces to the model in~\cite{qiu2015minimizing} and, as it is shown later, our proposed algorithm will respect the matching constraints, therefore, it is compatible with both models.

\subsection*{Coflow}
A coflow is a collection of flows whose completion time is determined by the completion time of the latest flow in the collection. The coflow $k$ can be denoted as an $N \times N$ demand matrix $D^{(k)}$. Every flow is a triple $(i,j,k)$, where $i \in \calI$ is its source node, $j \in \calJ$ is its destination node, and $k$ is the coflow to which it belongs. The size of flow $(i,j,k)$ is denoted by $d_{ij}^k$, which is the $(i,j)$-th element of the matrix $D^{(k)}$.
For simplicity, we assume that all flows within a coflow arrive to the system at the same time (as in~\cite{qiu2015minimizing}); however, our results still hold for the case that flows of a coflow are released at different times (which could indeed happen in practice~\cite{chowdhury2015efficient}). A $3 \times 3$ switch architecture is shown in Figure~\ref{network} as an example, where a coflow is illustrated by means of input queues, e.g., the file in the $j$-th queue at the source link $i$ indicates that the coflow has a flow from source server $i$ to destination server $j$. For instance, in Figure~\ref{network}, the illustrated coflow has $7$ flows in total, while two of its flows have source server $1$, one goes to destination server $1$ and the other to destination server $3$.
\begin{figure}[t]
		\centering
		\includegraphics[trim={0 1.1in 0 1.5in}, clip,width=2.6 in, height=1.2 in]{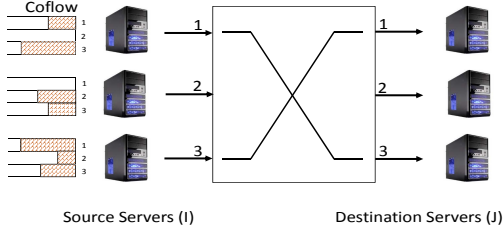}
		\caption{A coflow in a $3 \times 3$ switch architecture.}
		\label{network}
\end{figure}
	
\subsection*{Total Weighted Coflow Complettion Time}
We consider the coflow scheduling problem with release dates. There is a set of $K$ coflows denoted by $\calK$. Coflow $k\in \calK$ is released (arrives) at time $r_k$ which means it can only be scheduled after time $r_k$. We use $f_k$ to denote the finishing (completion) time of coflow $k$, which, by definition of coflow, is the time when all its flows have finished processing. In other words, for every coflow $k \in \calK$,
\begin{equation}
\label{cct}
f_k =\max_{i \in \calI, j \in \calJ}f_{i j}^k,
\end{equation}
where $f_{i j}^k$ is the completion time of flow $(i,j,k)$.

For given positive weights $w_k$, $k \in \calK$, the goal is to minimize the weighted sum of coflow completion times: $\sum_{k=1}^K w_k f_k$. The weights can capture different priority for different coflows. In the special case that all the weights are equal, the problem is equivalent to minimizing the average coflow completion time.

Define
\be \label{eq:T}
T=\max_{k \in \calK} r_k+ \sum_{k \in \calK}\sum_{i \in \calI}\sum_{j \in \calJ} d_{i j}^k.
\ee
Note that $T$ is clearly an upper bound on the minimum time required for processing of all the coflows. We denote by $x_{ij}^k(t)$ the transmission rate assigned to flow $(i,j,k)$ at time $t \in [0,T]$.
Then the optimal rate control must solve the following minimization problem
%which we refer to as TWCCT (Total Weighted Coflow Completion Time):
%over decision variables $\big \{x_{ij}^k(t) \big \}_{i,j=1}^N$ for $k=1,...,K$, and $t \in [0,T]$, flow completion times $\big \{f_{ij}^k \big \}_{i,j=1}^N$ for $k=1,...,K$, and coflow completion times $\big \{f_k \big \}_{k=1}^K$,
\begin{subequations}
\label{optimization}
\begin{align}
\mbox{minimize} \ & \sum_{k=1}^K w_k f_k\\
\label{rcct}
\text{subject to: }& f_k \geq f_{i j}^k, \ \ i \in \calI, j \in \calJ, k \in \calK\\
\label{demand}
& d_{ij}^k=\int_0^{f_{ij}^k} x_{ij}^k (t) dt, \ \ i \in \calI, j \in \calJ, k \in \calK\\
\label{capcons1}
& \sum_j \sum_k x_{ij}^k (t) \leq 1, \ \ i \in \calI,  t \in [0,T]\\
\label{capcons2}
& \sum_i \sum_k x_{ij}^k (t) \leq 1, \ \ j \in \calJ, t \in [0,T]\\
\label{releasedate}
& x_{ij}^k(t)=0, \ \ \forall t<r_k, \ i \in \calI, j \in \calJ, k \in \calK \\
\label{nonneg}
& x_{ij}^k(t) \geq 0, \ \ i \in \calI, j \in \calJ, k \in \calK,  t \in [0,T]
\end{align}
\end{subequations}
In the above, the constraint~\dref{rcct} indicates that each coflow $k$ is completed when all its flows have been completed. Note that since the optimization~\dref{optimization} is a minimization problem, a coflow completion time is equal to the completion time of its latest flow, in agreement with~\dref{cct}. The constraint~\dref{demand} ensures that the demand (file size) of every flow, $d_{ij}^k$, is transmitted by its completion time, $f_{ij}^k$.
%Note that since variable $f_{ij}^k$ is appeared in upper limit of the integral on rate allocation $x_{ij}^k$, this formulation of the problem is not linear.
Constraints~\dref{capcons1} and~\dref{capcons2} state the capacity constraints on source links and destination links, respectively. The fact that a flow cannot be transmitted before its release date (which is equal to release date of its corresponding coflow) is captured by the constraint~\dref{releasedate}. Finally, the constraint~\dref{nonneg} simply states that the rates are non-negative. 
\section{Motivations and Challenges}
\label{Overview}
The coflows can be widely different in terms of the number of parallel flows, the size of individual flows, the groups of servers involved, etc. Heuristics from traditional flow/task scheduling, such as shortest- or smallest-first policies~\cite{pinedo2015scheduling,schrage1968letter}, do not have a clear equivalence in coflow scheduling. One can define a shortest or smallest-first policy based on the number of parallel flows in a coflow, or the aggregate flow sizes in a coflow, however these policies perform poorly~\cite{chowdhury2014efficient}, as they do not completely take all the characteristics of coflows into consideration.

Recall that the completion time of a coflow is dominated by its slowest flow (as described by \dref{cct} or \dref{rcct}). Hence, it makes sense to slow down all the flows in a coflow to match the
completion time of the flow that will take the longest to finish. The unused capacity then can be used to allow other coexisting coflows to make progress and the total (or average) coflow completion time decreases. Varys~\cite{chowdhury2014efficient} is the first heuristic that effectively implements this intuition by combining Smallest-Effective-Bottleneck-First and Minimum-Allocation-for-Desired-Duration policies. Before describing Varys, we present a few definitions that are used in the rest of this paper.

\begin{definition}[Aggregate Size and Effective Size of a Coflow]
\label{definition1}
Let
\begin{equation}
\label{restrictedsize}
d_i^k= \sum_{j \in \calJ} d_{ij}^k;\ \ d_j^k= \sum_{i \in \calI} d_{ij}^k,
\end{equation}
be respectively the aggregate flow size that coflow $k$ needs to send from source node $i$ and receive at destination node $j$.
The effective size of coflow $k$ is defined as
\begin{equation}
\label{coflowload}
W(k)=\max \{\max_{i \in \calI} d_i^k , \max_{j \in \calJ} d_j^k\}.
\end{equation}
\end{definition}
Thus $W(k)$ is the maximum amount of data that needs to be sent or received by a node for coflow $k$. Note that, due to normalized capacity constraints on links, \textit{when coflow $k$ is released, we need at least $W(k)$ amount of time to process all its flows}.

\textbf{Overview of Varys.}
Varys~\cite{chowdhury2014efficient} orders coflows in a list based on their effective size in an increasing order. Transmission rates of individual flows of the first coflow in the list are set such that all its flows complete at the same time. The remaining capacity of links are updated and iteratively distributed among other coflows in the list in a similar fashion. Formally, the completion time of coflow $k$, $k=1,...,K$, is calculated as follows
\begin{equation*}
\Gamma^k=\max\{\max_{i \in \calI} \frac{d_i^k}{Rem(i)}, \max_{j \in \calJ} \frac{d_j^k}{Rem(j)} \},
\end{equation*}
where $Rem(i)$ (similarly, $Rem(j)$) is the remaining capacity of input link $i$ (output link $j$) after transmission rates of all coflows $k^\prime < k$ are set. Then for flow $(i,j,k)$, Varys assigns transmission rate $x_{ij}^k=d_{ij}^k/\Gamma^k$. In case that there is still idle capacity, for each input link $i \in \calI$, the remaining capacity is allocated to the flows of coflows subject to capacity constraints in corresponding output links. Once the first coflow completes, all the flow sizes and the scheduling list are updated and the iterative procedure is repeated to complete the second coflow and distribute the unused capacity. The procedure is stopped when all the coflows are processed.

While Varys performs better than traditional flow scheduling algorithms, it could still be inefficient. The main reason is that \textit{Varys is oblivious to the dependency among coflows who share a (source or destination) node}. To further expose this issue, we present a simple example.
\begin{example}[Inefficiency of Varys]
Consider the $2 \times 2$ switch network illustrated in Figure~\ref{example} where there are $3$ coflows in the system. In Figure~\ref{example1}, the effective coflow sizes are $W(1)=W(2)=W(3)=1$, therefore, Varys cannot differentiate among coflows. Scheduling coflows in the order $\{1,2,3\}$ or $\{2,3,1\}$ are both possible under Varys but they result in different total completion times, $1+2+2=5$ and $1+1+2=4$, respectively (assuming the weights are all one for all the coflows). Next, consider a slight modification of flow sizes, as shown in Figure~\ref{example2}. In this example $W(1)=2$ and $W(2)=W(3)=3$. Based on Varys algorithm, coflow $1$ is scheduled first during time interval $(0,2]$ at rate $1$. When coflow $1$ completes, coflows $2$ and $3$ are scheduled in time interval $(2,5]$; hence, the total completion time will be $2+5+5=12$. However, if we schedule coflows $2$ and $3$ first, the total completion times will reduce to $3+3+5=11$. Note that in both examples, coflow $1$ completely blocks coflows $2$ and $3$, which is not captured by Varys. In fact, the negative impact of ignoring configuration of coflows and their shared nodes is much more profound in large networks with a large number of coflows (see simulations in Section~\ref{simulation}).
 \begin{figure}[t]
 	\centering
 	\begin{subfigure}[t]{1\columnwidth}
 		\centering
 		\includegraphics[width=\textwidth]{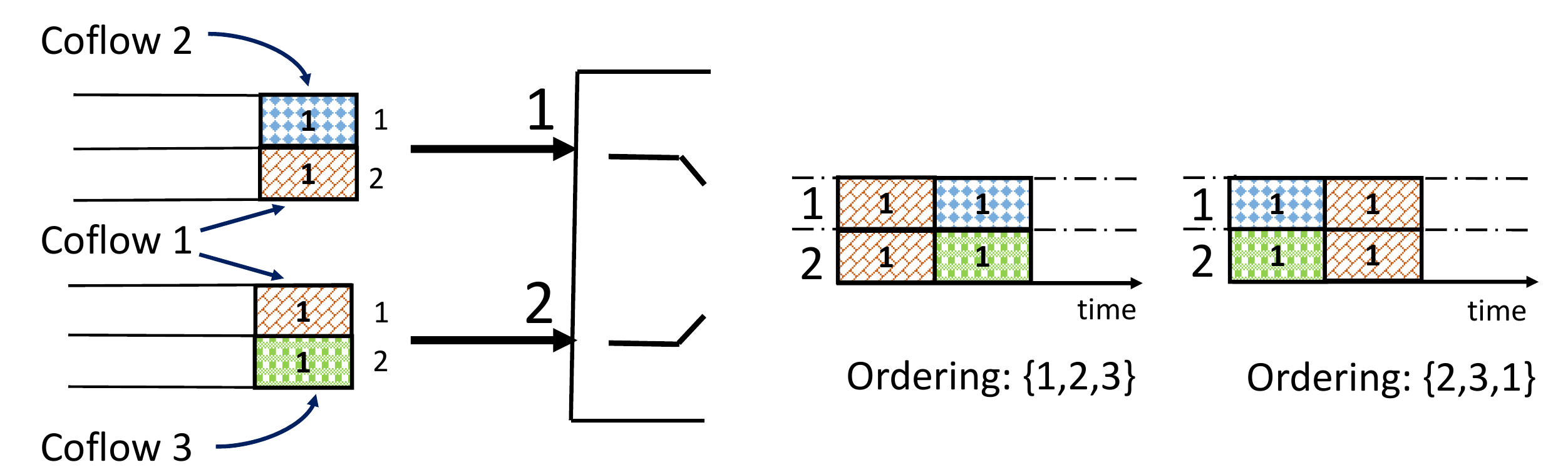}%
 		\caption{All coflows have equal effective size. Both orderings are possible under Varys, with the total completion time of $1+2+2=5$ and $1+1+2=4$, for the left and right ordering respectively.}
 		\label{example1}%
 	\end{subfigure}\hfill
% 	\vspace{-.1in}
 	\begin{subfigure}[t]{1\columnwidth}
 		\centering
 		\includegraphics[width=\textwidth]{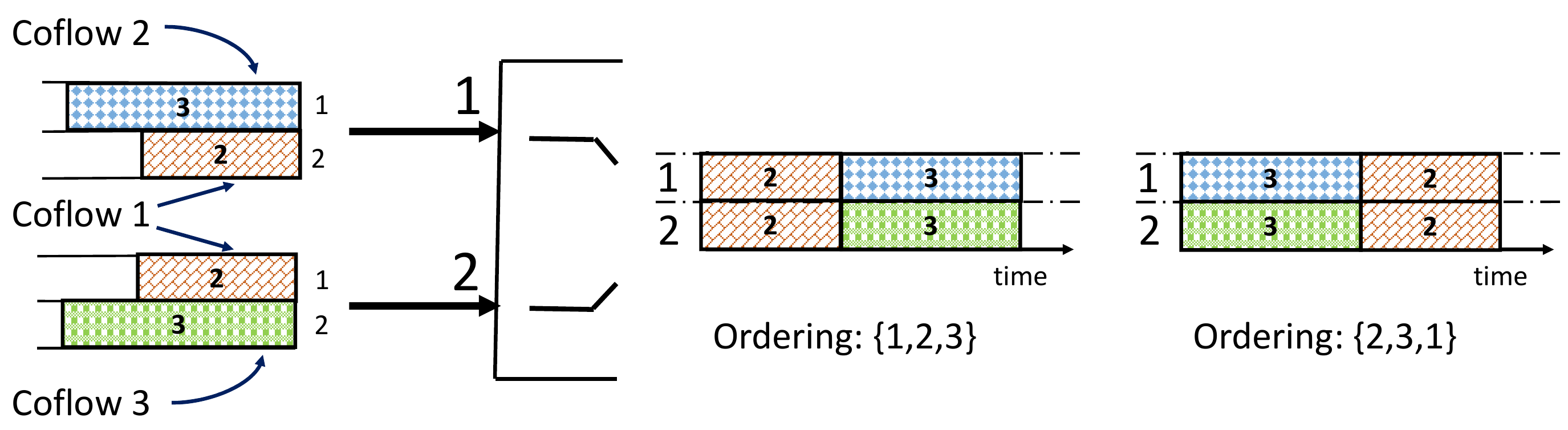}%
 		\caption{Varys schedules coflow $1$ first, according to the ordering $\{1,2,3\}$, which gives a total completion time of $2+5+5=12$. The optimal schedule is the ordering $\{2,3,1\}$ with a total completion time of $3+3+5=11$.}%
 		\label{example2}%
 	\end{subfigure}
 	\caption{Inefficiency of Varys in a $2 \times 2$ switch network with $3$ coflows.}
 	\label{example}
 	\vspace{-.2in}
 \end{figure}
\end{example}

\textbf{Overview of LP-based algorithms.} The papers~\cite{qiu2015minimizing} and~\cite{oursigmetricswork} use Linear Programs (LPs) (based on interval-indexed variables or ordering variables) that capture more information about coflows and provide a better ordering of coflows for scheduling compared to Varys~\cite{chowdhury2014efficient}. At the high level, the technical approach in these papers is based on partitioning jobs (coflows) into polynomial number of groups based on solution to a polynomial-sized relaxed linear program, and minimizing the completion time of each group by treating the group as a single coflow. Grouping can have a significant impact on decreasing the completion time of coflows. For instance, in view of examples in Figure~\ref{example}, grouping coflows $2$ and $3$, and scheduling them first, decreases the total completion time as explained.

\textbf{NP-hardness and connection to the concurrent open shop problem.} The concurrent open shop problem~\cite{mastrolilli2010minimizing} can be essentially viewed as a special case of the coflow scheduling problem \emph{when demand matrices are diagonal} (in the jargon of concurrent open shop problem, the coflows are jobs, the flows in each coflow are tasks for that job, and the destination nodes are machines with unit capacities). It is known that it is NP-complete to approximate the concurrent open shop problem, when jobs are released at time zero, within a factor better than $2-\epsilon$ for any $\epsilon>0$~\cite{sachdeva2013optimal}. Although the model we consider for coflow scheduling is different from the model used in~\cite{qiu2015minimizing}, similar reduction as proposed in~\cite{qiu2015minimizing} can be leveraged to show NP-completness of the coflow scheduling problem. More precisely, every instance of the concurrent open shop problem can be reduced to an instance of coflow scheduling problem (see Appendix~\ref{2appexp} for the details), hence it is NP-complete to approximate the coflow scheduling problem (without release dates) within $2-\epsilon$, for any $\epsilon>0$. There are $2$-approximation algorithms for the concurrent open shop (e.g.,~\cite{mastrolilli2010minimizing}), however, these algorithms cannot be ported to the coflow scheduling problem due to the coupling of source and destination link capacity constraints in the coflow scheduling problem (see Appendix~\ref{2appexp} for a counter example).

Next, we describe our coflow scheduling algorithm. The algorithm is based on a linear program formulation for sorting the coflows followed by a simple list scheduling policy
%nough to provide an efficient order of coflows, and actual scheduling of flows of coflows is also of great importance to achieve good performance. Neglect of this fact yields the claimed $2$-approximation algorithm in~\cite{luo2016towards}, for coflow scheduling when all the release times are zero, to be flawed. This algorithm relies on the well-known $2$-approximation algorithm for concurrent open shop problem~\cite{mastrolilli2010minimizing} to obtain an ordering for coflows combined with an assumption for scheduling coflows which does not hold in general, as we show here.
%The assumption is the following: given the ordering of $K$ coflows, there always exists a schedule in which the first coflow completes at time $W(1)$, the second coflow completes at time $W(1,2)$, and so on, till the last coflow completes at time $W(1,\cdots,K)$. We provide a counter example to show that such schedule does not always exist.
%\input{counterexample} 
\section{Linear Programing (LP) Relaxation}
\label{LPRelaxation}
In this section, we use \textit{linear ordering variables} (see, e.g., \cite{potts1980algorithm,hall1996scheduling,gandhi2008improved,mastrolilli2010minimizing}) to present a relaxed integer program of the original scheduling problem~\dref{optimization}. We then relax these variables to obtain a linear program (LP). In the next section, we use the optimal solution to this LP as a subroutine in our deterministic algorithm.

\textbf{Ordering variables.} For each pair of coflows, if both coflows have \textit{some} flows incident at some node (either originated from or destined at that node), we define a binary variable which indicates which coflow finishes all its flows before the other coflow does so in the schedule. Formally, for any two coflows $k, k^\prime$ with aggregate flow sizes $d^k_m\neq 0$ and $d^{k\prime}_m\neq 0$ on some node $m \in \calI \cup \calJ$
%define $\calK_m$ as the set of coflows incident to node $m$, i.e.,
%\be
%\calK_m=\{k \in \calK: d^k_m\neq 0\},
%\ee
%where $d^k_m$ is the aggregate size of flows of coflow $k$ at node $m$, as
(recall definition \dref{restrictedsize}), we introduce a binary variable $\delta_{k k^\prime} \in \{0,1\}$ such that $\delta_{k k^\prime}=1$ if coflow $k$ finishes all its flows before coflow $k^\prime$ finishes all its flows, and it is $0$ otherwise. If both coflows finish their flows at the same time (which is possible in the case of continuous-time rate control), we set either one of $\delta_{k k^\prime}$ or $\delta_{k^\prime k}$ to $1$ and the other one to $0$, arbitrarily.

\textbf{Relaxed Integer Program (IP).} We formulate the following Integer Program (IP):
\begin{subequations}
\label{RIP}
\begin{align}
\label{ILobj}
\mathbf{(IP)}\ \min \ \ &\sum_{k=1}^K w_k f_k \\
\label{matching1}
& f_k \geq  d_i^k+ \sum_{k^\prime\in \calK}d_i^{k^\prime} \delta_{k^\prime k} \ \ i \in \calI, k \in \calK \\
\label{matching2}
& f_k \geq  d_j^k+\sum_{k^\prime\in \calK} d_j^{k^\prime} \delta_{k^\prime k} \ \ j \in \calJ, k \in \calK\\
\label{rt}
& f_k \geq W(k) + r_k \ \ k \in \calK\\
\label{prec1}
&\delta_{k k^\prime}+\delta_{k^\prime k}=1 \ \ k,k^\prime \in \calK\\
\label{int1}
&\delta_{k k^\prime} \in \{0,1\} \ \ k,k^\prime \in \calK
\end{align}
\end{subequations}
In the above, to simplify the formulation, we have defined $\delta_{k k^\prime}$, for all pairs of coflows, by defining $d_m^k=0$ if coflow $k$ has no flow originated from or destined to node $m$.

The constraint~\dref{matching1} (similarly~\dref{matching2}) follows from the definition of ordering variables and the fact that flows incident to a source node $i$ (a destination node $j$) are processed by a single link of unit capacity. To better see this, note that the total amount of traffic can be sent in the time period $(0,f_k]$ over the $i$-th link is at most $f_k$. This traffic is given by the right-hand-side of~\dref{matching1} (similarly~\dref{matching2}) which basically sums the aggregate size of coflows incident to node $i$ that finish their flows before coflow $k$ finishes its corresponding flows, plus the aggregate size of coflow $k$ at node $i$ itself, $d_i^k$. This implies constraint~\dref{matching1} and~\dref{matching2}.
The fact that each coflow cannot be completed before its release date plus its effective size is captured by constraint~\dref{rt}. The next constraint~\dref{prec1} indicates that for each two incident coflows, one precedes the other.

Note that this optimization problem is a relaxed integer program for the problem~(\ref{optimization}), since the set of constraints are not capturing all the requirements we need to meet for a feasible schedule. For example, we cannot start scheduling flows of a coflow when it is not released yet, while constraint~\dref{rt} does not necessarily avoid this, thus leading to a smaller value of finishing time compared to the optimal solution to \dref{optimization}. Further, release dates and scheduling constraints in optimization (\ref{optimization}) might cause idle times in flow transmission of a node, therefore yielding a larger value of finishing time for a coflow than what is restricted by \dref{matching1}, \dref{matching2}, \dref{rt}. To further illustrate this issue, we present a simple example.

\begin{example} Consider a $2 \times 2$ switch network as in Figure~\ref{network2}. Assume there are 4 coflows, each has one flow. Flow $(1,1,1)$ is released at time $0$ with size $1$, and the other three flows are released at time $1$ with size $2$. It is easy to check that the following values for the ordering variables and flow completion times satisfy all the constraints~\dref{matching1}$-$\dref{int1}. For brevity, we only report the ordering variables for coflows that actually share a node. For example, it is redundant to consider ordering variables corresponding to coflow $1$ and coflow $4$ as they are not incident at any (source/destination) node and any value for their associated pairwise ordering variables does not have any impact on the optimal value for IP \dref{RIP}. Below, the ordering variables and coflow completion times are presented, and all the ordering variables which are not specified can be taken as zero.
\begin{equation*}
\begin{aligned}
&\delta_{12}=1, \ \ &&\delta_{34}=1,\\
&\delta_{13}=1, \ \ &&\delta_{24}=1,\\
&f_1=1, \ \ &&f_2=3, \\
&f_3=3, \ \ &&f_4=4. \\
\end{aligned}
\end{equation*}
\begin{figure}[t]
	\centering
	\includegraphics[width=0.75\columnwidth, height=1in]{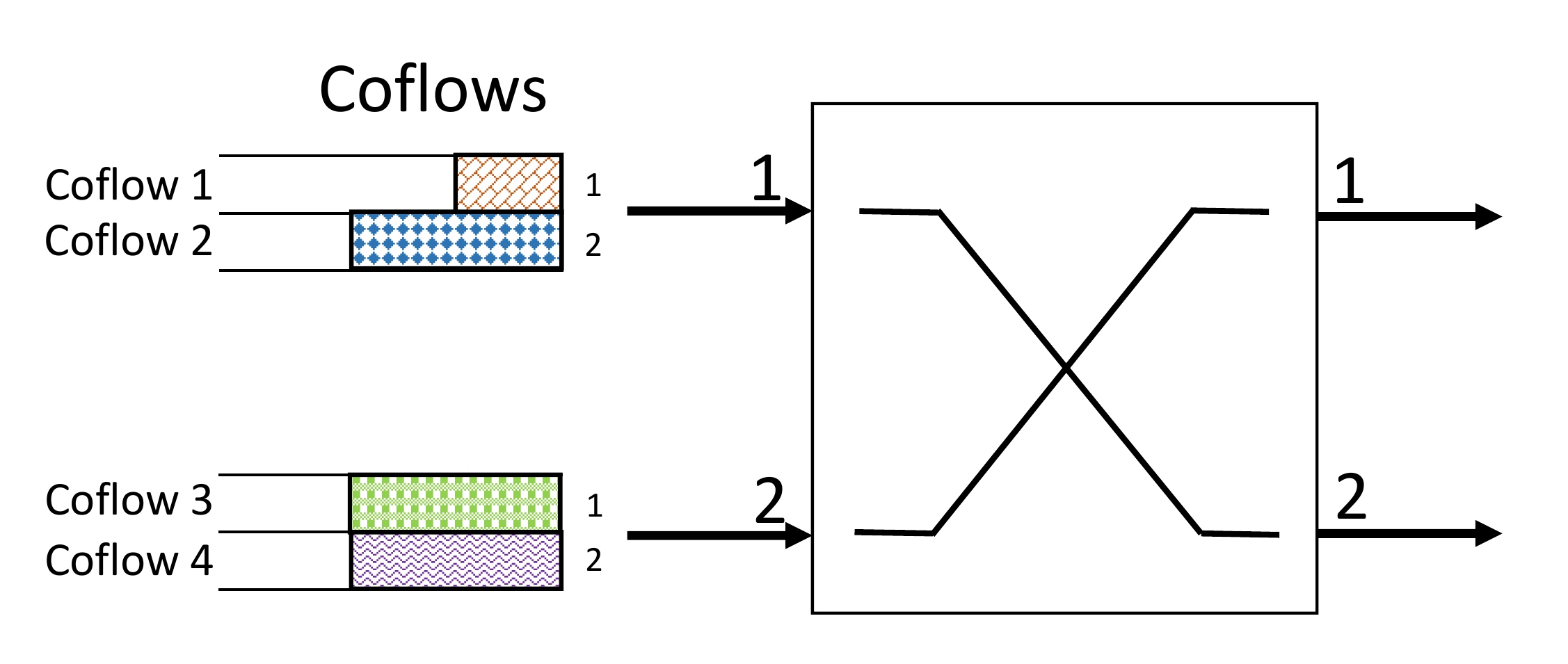}
	\caption{4 coflows in a $2 \times 2$ switch architecture, flow $(1,1)$ is released at time $0$, and all the others are released at time $1$.}
	\label{network2}
\vspace{-0.1in}
\end{figure}
While these values satisfy ~\dref{matching1}$-$\dref{int1}, this is not a valid schedule since it requires transmission of flow $(2,2,4)$ starting at time $0$, while it is not released yet. To see this, note that $f_1=1$, so to finish processing of coflow $1$ or equivalently flow $(1,1,1)$ by time $1$, we need to start its transmission at maximum rate at time $0$. Then, due to the capacity constraints, the first time flows $(1,2,2)$ and $(2,1,3)$ can start transmission is at time $1$, when flow $(1,1,1)$ has been completed. Since we require to complete both of these flows at time $3$, they need to be transmitted at maximum rate in the time interval $(1,3]$. Therefore, the only way to finish flow $(2,2,4)$ at time $4$ is to send one unit of its data in time interval $(0,1]$ and its remaining unit of data in time interval $(3,4]$, but this flow has not been released before time $1$. So the proposed IP does not address all the scheduling constraints.
\end{example}
\textbf{Relaxed Linear Program (LP).} In the linear program relaxation, we allow the ordering variables to be fractional. Specifically, we replace the constraints~\dref{int1} with the constraints \dref{fraction1} below. We refer to the obtained linear problem by (LP).
\begin{subequations}
\label{LP}
\begin{align}
\mathbf{(LP)}\ \min \ \ &\sum_{k=1}^K w_k f_k  \\
\text{subject to:}&\ \text{\dref{matching1} -- \dref{prec1}},\nonumber \\
\label{fraction1}
&\delta_{k k^\prime}  \in [0,1] \ \ k,k^\prime \in \calK
\end{align}
\end{subequations}
We use $\tilde{f}_k$ to denote the optimal solution to this LP for the completion time of coflow $k$. Also we use $\widetilde{\text{OPT}}=\sum_k w_k \tilde{f_k}$ to denote the corresponding objective value. Similarly we use $f_k^\star$ to denote the optimal completion time of coflow $k$ in the original coflow scheduling problem \dref{optimization}, and use $\text{OPT}=\sum_k w_k {f_k^\star}$ to denote its optimal objective value. The following lemma establishes a relation between $\widetilde{\text{OPT}}$ and $\text{OPT}$.
\begin{lemma}
	\label{optlb}
	The optimal value of the LP, $\widetilde{\text{OPT}}$, is a lower bound on the optimal total weighted completion time $\text{OPT}$ of coflow scheduling problem.
\end{lemma}
\begin{proof}
	Consider an optimal solution to the optimization problem~\dref{optimization}. We set ordering variables so as $\delta_{kk^\prime}=1$ if coflow $k$ precedes coflow $k^\prime$ in this solution, and $\delta_{kk^\prime}=0$, otherwise. If both coflows finish their corresponding flows at the same time, we set either one to $1$ and the other one to $0$. We note that this set of ordering variables and coflow completion times satisfies constraints~\dref{matching1} and~\dref{matching2} (by taking integral from both side of constraint~\dref{capcons1} and~\dref{capcons2} from time $0$ to $f_k$) and also constraint~\dref{rt} (by combining constraints~\dref{demand} and~\dref{releasedate}). Furthermore, the rest of (LP) constraints are satisfied by the construction of ordering variables. Therefore, optimal solution of problem~\dref{optimization} can be converted to a feasible solution to (LP). Hence, the optimal value of LP, $\widetilde{\text{OPT}}$, is at most equal to $\text{OPT}$.
\end{proof} 
\section{Approximation Algorithm}
\label{AppAlg}
In this section, we describe our polynomial-time coflow scheduling algorithm and state the main results about its performance guarantees.

The algorithm has three steps: (i) solve the relaxed LP~\dref{LP}, (ii) use the solution of the relaxed LP to order flows of coflows, and (iii) apply a simple list scheduling algorithm based on the ordering. The relaxed LP~\dref{LP} has $O(K^2)$ variables and $O(K^2+KN))$ constraints and can be solved efficiently in polynomial time, e.g. using interior point method~\cite{renegar1988polynomial}.

The approximation algorithm is depicted in Algorithm~\ref{Alg1}. We order coflows based on values of $\tilde{f}_k$ (optimal solution to LP) in nondecreasing order. More precisely, we re-index coflows such that,
\begin{equation}
\label{order}
\tilde{f}_1 \leq \tilde{f}_2 \leq ... \leq \tilde{f}_K.
\end{equation}
Ties are broken arbitrarily. The algorithm then maintains a list of flows such that for every two flows $(i,j,k)$ and $(i^\prime,j^\prime,k^\prime)$ with $k < k^\prime$, flow $(i,j,k)$ is placed before flow $(i^\prime,j^\prime,k^\prime)$ in the list. Flows of the same coflow are listed in an arbitrary order. The algorithm scans the list starting from the first flow and schedules a flow if both its corresponding source and destination links are idle at that time. Upon completion of a flow or arrival of a coflow, the algorithm preempts the schedule, updates the list, and starts scheduling the flows in the updated list.
\begin{algorithm}
	\caption{Deterministic Coflow Scheduling Algorithm}
	\label{Alg1}
	\begin{algorithmic}
	 \State Suppose coflows $\bigg \{d_{ij}^k \bigg \}_{i,j=1}^N$, $k \in \calK$, with release dates $r_k$, $k \in \calK$, and weights $w_k$, $k \in \calK$, are given.
	\begin{algorithmic}[1]
		\State Solve the linear program (LP) and denote its optimal solution by $\{\tilde{f_k}; k \in \calK\}$.
		\State Order and re-index the coflows such that:
		\begin{equation}
		\label{order2}
		\tilde{f}_1 \leq \tilde{f}_2 \leq ... \leq \tilde{f}_K,
		\end{equation}
		where ties are broken arbitrarily.
		\State Wait until the first coflow is released.
			\While {There is some incomplete flow,}
			\State List the released and incomplete flows respecting the ordering in~\dref{order2}. Let $L$ be the total number of flows in the list.
			\For {$l=1$ to $L$}
			\State Denote the $l$-th flow in the list by $(i_l,j_l,k_l)$,
				\If {Both the links $i_l$ and $j_l$ are idle,}
				\State Schedule flow $(i_l,j_l,k_l)$.
				\EndIf
			\EndFor
			\While {No flow is complete and no coflow is released}
			\State Transmit the flows that get scheduled in line $9$ with rate $1$.
			\EndWhile
			\EndWhile

	\end{algorithmic}
	\end{algorithmic}
\end{algorithm}

The main result regarding the performance of Algorithm~\ref{Alg1} is stated in Theorem~\ref{detalgper}.
\begin{theorem}
	\label{detalgper}
	Algorithm~\ref{Alg1} is a deterministic $5$-approximation algorithm for the problem of minimizing total weighted completion time of coflows with release dates.
\end{theorem}
When all coflows are released at time $0$, we can improve the algorithm performance ratio.
\begin{corollary}
	\label{cor1}
	If all coflows are released at time $0$, then Algorithm~\ref{Alg1} is a deterministic $4$-approximation algorithm.
\end{corollary}

\section{Proof Sketch of Main Results}\label{proofs}
In this section, we present the sketch of proofs of the main results for our polynomial-time coflow scheduling algorithm. Before proceeding with the proofs, we make the following definitions.
\begin{definition}[Aggregate Size and Effective Size of a List of Coflows]\label{def2}
For a list of $K$ coflows and for a node $s \in \calI \cup \calJ$, we define $W(1,\cdots,k;s)$ to be the amount of data needs to be sent or received by node $s$ in the network considering only the first $k$ coflows. We also denote by $W(1,\cdots,k)$ the effective size of the aggregate coflow constructed by the first $k$ coflows, $k\leq K$. Specifically,
\begin{equation}
\label{nodeload}
W(1,\cdots,k;s) = \sum_{l=1}^k d_{s}^l
\end{equation}
\begin{equation}
\label{load}
W(1,\cdots,k)=\max_{s \in \calI \cup \calJ} W(1,\cdots,k;s)
\end{equation}
\end{definition}
\subsection{Bounded completion time for the collection of coflows}
Consider the list of coflows according to the ordering in~\dref{order} and define $W(1,\cdots,k)$ based on Definition~\ref{def2}.
The following lemma demonstrates a relationship between completion time of coflow $k$ obtained from (LP) and $W(1,\cdots,k)$ which is used later in the proofs.
\begin{lemma}
	\label{bound1}
	$\tilde{f}_k \geq \frac{W(1,\cdots,k)}{2}$.
\end{lemma}
\begin{proof}
	The proof uses similar ideas as in Gandhi, et al.~\cite{gandhi2008improved} and Kim~\cite{kim2005data}. Using constraint~\dref{matching1}, for any source node $i \in \calI$, we have
	\begin{equation}
	\begin{aligned}
	d_{i}^l \tilde{f_l} \geq  (d_{i}^l)^2+\sum_{l^\prime \in \calK} d_{i}^l d_{i}^{l^\prime} \delta_{l^\prime l}
	\end{aligned}
	\end{equation}
	which implies that,
	\begin{equation}
	\begin{aligned}
	\label{help}
	\sum_{l=1}^k d_{i}^l \tilde{f_l} \geq & \sum_{l=1}^k (d_i^l)^2 +\sum_{l=1}^k \sum_{l^\prime=1}^k d_i^{l^\prime} d_i^l \delta_{l^\prime l}\\
	=& \frac{1}{2} \bigg ( 2 \times \sum_{l=1}^k (d_i^l)^2\\
	&+ \sum_{l=1}^k \sum_{l^\prime=1}^k \big (d_i^{l^\prime} d_i^l \delta_{l^\prime l}+d_i^{l^\prime} d_{i}^l \delta_{l l^\prime} \big ) \bigg )
	\end{aligned}
	\end{equation}
	We simplify the right-hand side of~\dref{help}, using constraint~\dref{prec1}, combined with the following equality
	\begin{equation}
	\sum_{l=1}^k (d_{i}^l)^2+ \sum_{l=1}^k \sum_{l^\prime=1}^k d_i^{l^\prime} d_i^l =(\sum_{l=1}^k d_i^l)^2,
	\end{equation}
	and conclude that
	\begin{equation}
	\begin{aligned}
	\sum_{l=1}^k d_{i}^l \tilde{f_l} \geq & \frac{1}{2} \sum_{l=1}^k (d_{i}^l)^2 + \frac{1}{2} (\sum_{l=1}^k d_{i}^l)^2 \\
	\geq & \frac{1}{2} (\sum_{l=1}^k d_{i}^l)^2 =\frac{1}{2} (W(1,\cdots,k;i))^2
	\end{aligned}
	\end{equation}
	Where the last equality follows from Definition~\ref{nodeload}.
	Similar argument results in the following inequality for any destination node $j \in \calJ$, i.e.,
	\begin{equation*}
	\sum_{l=1}^k d_{j}^l \tilde{f_l} \geq \frac{1}{2} (W(1,\cdots,k;j))^2.
	\end{equation*}
	Now consider the node $s^\star$ which has the maximum load induced by the first $k$ coflows, namely, $W(1,\cdots,k)=W(1,\cdots,k;s^\star)$.
	\begin{equation}
	\begin{aligned}
	\tilde{f_k} W(1,\cdots,k;s^\star)&= \tilde{f_k}\sum_{l=1}^k d_{s^\star}^l\\
	& \geq \sum_{l=1}^k d_{s^\star}^l \tilde{f_l}\\
	& \geq \frac{1}{2} (W(1,\cdots,k;s^\star))^2
    \end{aligned}
	\end{equation}
	This implies that,
    \begin{equation}
    \tilde{f_k} \geq \frac{1}{2} W(1,\cdots,k;s^\star)=\frac{1}{2} W(1,\cdots,k).
    \end{equation}
This completes the proof.
\end{proof}

Note that $W(1,\cdots,k)$ is a lower bound on the time that it takes for all the first $k$ coflows to be completed (as a result of the capacity constraints in the optimization~\dref{optimization}). Hence, Lemma~\ref{bound1} states that by allowing ordering variables to be fractional, completion time of coflow $k$ obtained from (LP) is still lower bounded by half of $W(1,\cdots,k)$.
\subsection{Proof of Theorem~\ref{detalgper} and Corollary~\ref{cor1}}
\begin{proof}[Proof of Theorem~\ref{detalgper}]
Recall that we use $\hat{f}_k$ to denote the actual coflow completion times under our deterministic algorithm. Suppose flow $(i,j,k)$ is the last flow of coflow $k$ that is completed. In general, Algorithm~\ref{Alg1} may preempt a flow several times during its execution. For now, suppose flow $(i,j,k)$ is not preempted and use $t_k$ to denote the time when its transmission is started (the arguments can be easily extended to the preemption case as we show at the end of the proof). Therefore
\begin{equation}
\label{eq2}
\hat{f_k}=\hat{f}_{ij}^k=t_k+d_{ij}^k
\end{equation}
From the algorithm description, $t_k$ is the first time that both links $i$ and $j$ are idle and there are no higher priority flows to be scheduled (i.e., there is no flow $(i,j,k^\prime)$ from $i$ to $j$ with $k^\prime <k$ in the list). By definition of $W(1,\cdots,k;s)$, node $s$, $s\in \{i,j\}$, has $W(1,\cdots,k;s)-d_{ij}^k$ data units to send or receive by time $t_k$. Since the capacity of all links are normalized to $1$, it should hold that
\begin{equation*}
\begin{aligned}
t_k &\leq r_k+W(1,\cdots,k;i)-d_{ij}^k+W(1,\cdots,k;j)-d_{ij}^k\\
&\leq r_k+2W(1,\cdots,k)-2d_{ij},
\end{aligned}
\end{equation*}
where the last inequality is by Definition~\ref{load}.
Combining this inequality with equality~\dref{eq2} yields the following bound on $\hat{f_k}$.
\begin{equation*}
\hat{f_k} \leq r_k+2 W(1,\cdots,k)
\end{equation*}
Using Lemma~\ref{bound1} and constraint~\dref{rt}, we can conclude that
\begin{equation*}
\hat{f_k} \leq 5 \tilde{f}_k,
\end{equation*}
which implies that
\begin{equation*}
\sum_{k=1}^K w_k \hat{f_k} \leq 5\sum_{k=1}^K w_k \tilde{f}_k.
\end{equation*}
This shows an approximation ratio of $5$ for Algorithm~\ref{Alg1} using Lemma~\ref{optlb}. Finally, if flow $(i,j,k)$ is preempted, the above argument can still be used by letting $t_k$ to be the starting time of its last piece and $d_{ij}^k$ to be the remaining size of its last piece at time $t_k$. This completes the proof.	
\end{proof}

\begin{proof}[Proof of Corollary~\ref{cor1}]
When all coflows are released at time $0$, $t_k \leq W(1,\cdots,k)-d_{ij}^k+W(1,\cdots,k)-d_{ij}^k$. The rest of the argument
is similar to the proof of Theorem~\ref{detalgper}. Therefore, the algorithm has approximation ratio of $4$ when all coflows are release at time $0$.
\end{proof} 
\section{Extension to Online Algorithm}
\label{online}
Similar to previous work~\cite{qiu2015minimizing,khuller2016brief}, Algorithm~\ref{Alg1} is an offline algorithm, and requires the complete knowledge of the flow sizes and release times. While this knowledge can be learned in long running services, developing online algorithms that deal with the dynamic nature and unavailability of this information is of practical importance. One natural extension of our algorithm to an online setting, assuming that the coflow information revealed at its release time, is as follows: Upon each coflow arrival, we re-order the coflows by re-solving the (LP) using the remaining coflow sizes and the newly arrived coflow, and update the list. Given the updated list, the scheduling is done as in Algorithm~\ref{Alg1}. To reduce complexity of the online algorithm, we may re-solve the LP once in every $T$ seconds, for some $T$ that can be tuned, and update the list accordingly. We leave theoretical and experimental study of this online algorithm as a future work. 
\section{Empirical Evaluations}
\label{simulation}
In this section, we present our simulation results and evaluate the performance of our algorithm for both cases of with and without release dates, under both synthetic and real traffic traces.
We also simulate the deterministic algorithms proposed in~\cite{oursigmetricswork,qiu2015minimizing} and Varys~\cite{chowdhury2014efficient} and compare their performance with the performance of our algorithm.
\subsection{Workload}
\label{wldescription}
We evaluate algorithms under both synthetic and real traffic traces.

\textbf{Synthetic traffic}: To generate synthetic traces we slightly modify the model used in~\cite{qiu2016experimental}. We consider the problem of size $K=160$ coflows in a switch network with $N=16$ input and output links. We denote by $M$ the number of non-zero flows in each coflow. We consider two cases:
\begin{itemize}[leftmargin=*]
		\item Dense instance: For each coflow, $M$ is chosen uniformly from the set $\{N, N+1, ..., N^2\}$. Therefore, coflows have $O(N^2)$ non-zero flows on average.
		\item Combined instance: Each coflow is sparse or dense with probability $1/2$. For each sparse coflow, $M$ is chosen uniformly from the set $\{1,2,...,N\}$, and for each dense coflow $M$ is chosen uniformly from the set $\{N, N+1, ..., N^2\}$.
\end{itemize}
Given the number $M$ of flows in each coflow, $M$ pairs of input and output links are chosen randomly. For each pair that is selected, an integer flow size (processing requirement) $d_{ij}$ is randomly selected from the uniform distribution on $\{1, 2,..., 100\}$. For the case of scheduling with release dates, we generate the coflow inter-arrival times uniformly from $[1, 100]$. We generate $100$ instances for each case and report the average algorithms' performance.

\textbf{Real traffic}: This workload was also used in~\cite{chowdhury2014efficient, qiu2015minimizing, oursigmetricswork}. The workload is based on a Hive/MapReduce trace at Facebook that was collected from a 3000-machine cluster with $150$ racks. In this trace, the following information is provided for each coflow: arrival time of the coflow in millisecond, locations of mappers (rack number to which they belong), locations of reducers (rack number to which they belong), and the amount of shuffle data in Megabytes for each reducer. We assume that shuffle data of each reducer in a coflow is evenly generated from all mappers specified for that coflow. The data trace consists of $526$ coflows in total from very sparse coflows (the most sparse coflow has only $1$ flow) to very dense coflows (the most dense coflow has $21170$ flows.). Similar to~\cite{qiu2015minimizing}, we filter the coflows based on the number of their non-zero flows, $M$. Apart from considering all coflows ($M\geq 1$), we consider three coflow collections filtered by the conditions $M \geq 10$, $M \geq 30$, and $M \geq 50$. In other words, we use the following 4 collections:
\begin{itemize}[leftmargin=*]
	\item All coflows,
	\item Coflows with $M \geq 10$,
	\item Coflows with $M \geq 30$,
	\item Coflows with $M \geq 50$.
\end{itemize}
Furthermore, the original cluster had a $10 : 1$ core-to-rack oversubscription ratio with a total bisection bandwidth of $300$ Gbps. Hence, each link has a capacity of $128$ MBps. To obtain the same traffic intensity offered to our network (without oversubscription), for the case of scheduling coflows with release dates, we need to scale down the arrival times of coflows by $10$. For the case of without release dates, we assume that all coflows arrive at time $0$.
\subsection{Algorithms}\label{sec:algorithms}
We simulate four algorithms: the algorithm proposed in this paper, Varys~\cite{chowdhury2014efficient}, the deterministic algorithm in~\cite{qiu2015minimizing}, and the deterministic algorithm in~\cite{oursigmetricswork}. We briefly overview these algorithms and also elaborate on the backfilling strategy that has been combined with the deterministic algorithms in~\cite{qiu2015minimizing, oursigmetricswork} to avoid under utilization of network resources.
	
\textbf{1. Varys~\cite{chowdhury2014efficient}}: Scheduling and rate assignments under Varys were explained in detail in Section~\ref{Overview}. There is a parameter $\delta$ in the original design of Varys that controls the tradeoff between fairness and completion time. Since we focus on minimizing the total completion time of coflows, we set $\delta$ to $0$ which yields the best performance of Varys. In this case, upon arrival or completion of a coflow, the coflow ordering is updated and the rate assignment is done iteratively as described in Section~\ref{Overview}.

\textbf{2. Interval-Indexed-Grouping~\cite{qiu2015minimizing}}: The algorithm requires discrete time (i.e., time slots) and is based on an interval-indexed formulation of a polynomial-time linear program (LP) as follows. The time is divided into geometrically increasing intervals. The binary decision variables $x_{lk}$ are introduced which indicate whether coflow $k$ is scheduled to complete within the $l$-th interval $(t_l,t_{l+1}]$. Using these binary variables, a lower bound on the objective function is formulated subject to link capacity constraints and the release date constraints. The binary variables are then relaxed leading to an LP whose solution is used for ordering coflows. More precisely, the relaxed completion time of coflow $k$ is defined as
$
f_k=\sum_l t_l x_{lk},
$
where $t_l$ is the left point of the $l$-th interval and $x_{lk} \in [0,1]$ is the relaxed decision variable. Based on the optimal solution to this LP, coflows are listed in an increasing order of their relaxed completion time. For each coflow $k$ in the list, $k=1,...,K$, we compute effective size of the cumulated first $k$ coflows in the list, $W(1,\cdots,k)$. All coflows that fall within the same time interval according to value of $W(1,\cdots,k)$ are grouped together and treated as a single coflow and scheduled so as to minimize its completion time.
Scheduling of coflows within a group makes use of the Birkhoff-von Neumann decomposition. If two data units from coflows $k $ and $k^\prime$ within the same group use the same pair of input and output, and $k$ is ordered before $k^\prime$, then we always process the data unit from coflow $k$ first. For backfilling, when we use a schedule that matches input $i$ to output $j$, if there is no more service requirement on the pair of input $i$ and output $j$ for some coflow in the current partition, we backfill in order from the flows on the same pair of ports in the subsequent coflows. We would like to emphasize that this algorithm needs to discretize time and is based on matching source nodes to destination nodes. We select the time unit to be $1/128$ second as suggested in~\cite{qiu2015minimizing} so that each port has a capacity of $1$ MB per time unit. We refer to this algorithm as `\underline{LP-II-GB}', where II stands for Interval-Indexed, and GB stands for Grouping and Backfilling.

\textbf{3. Ordering-Variable-Grouping~\cite{oursigmetricswork}}: We implement the deterministic algorithm in~\cite{oursigmetricswork}. Linear programming formulation is the same as LP in \dref{LP}. Coflows are then grouped based on the optimal solution to the LP. To schedule coflows of each group, we construct a single aggregate coflow denote by $D$ and schedule its flows to optimize its completion time. We assign transmission rate $x_{ij} = d_{ij}/W(D)$ to the flow from source node $i$ to destination node $j$ until its completion. Moreover, the continuous backfilling is done as follows: After assigning rates to aggregate coflow, we increase $x_{ij}$ until either capacity of link $i$ or link $j$ is fully utilized. We continue until for any node, either source or destination node, the summation of rates sum to one. We also transmit flows respecting coflow order inside of each partition. When there is no more service requirement on the pair of input $i$ and output $j$ for coflows of current partition, we backfill (transmit) in order from the flows on the same pair of ports from the subsequent coflows. We refer to this algorithm as `\underline{LP-OV-GB}', where OV stands for ordering variables, and GB stands for Grouping and Backfilling.

\textbf{4. Algorithm~\ref{Alg1}}: We implement our algorithm as described in Algorithm~\ref{Alg1}, and refer to it as `\underline{LP-OV-LS}', where LS stands for list scheduling.
\subsection{Evaluation Results}
\textbf{Performance of Our Algorithm.}
We report the ratios of total weighted completion time obtained from Algorithm~\ref{Alg1} and the optimal value of relaxed linear program~\dref{LP} (which is a lower bound on the optimal value of the coflow scheduling problem) to verify Theorem~\ref{detalgper} and Corollary~\ref{cor1}.
We only present results of the simulations using the real traffic trace, with equal weights and random weights. For the case of random weights, the weight of each coflow is chosen uniformly at random from the interval $[0,1]$. The results are more or less similar for other collections and for synthetic traffic traces and all are consistent with our theoretical results.

Table~\ref{ratio} shows the performance ratio of the deterministic algorithm for the cases of with and without release dates. All performances are within our theoretical results indicating the approximation ratio of at most $4$ when all coflows release at time $0$ and at most $5$ when coflows have general release dates. In fact, the approximation ratios for the real traffic trace are much smaller than $4$ and $5$ and very close to $1$.\\
\begin{table}
	\caption{Performance Ratio of Algorithm~\ref{Alg1}}
	\centering
	\begin{tabular}{c c c} % centered columns (4 columns)
		\hline\hline %inserts double horizontal lines
		Case & Equal weights & Random weights \\ [0.5ex] % inserts table
		%heading
		\hline % inserts single horizontal line
		Without release dates & $1.05$ & $1.06$ \\ % inserting body of the table
		With release dates & $1.034$ & $1.038$ \\ % [1ex] adds vertical space
		\hline %inserts single line
	\end{tabular}
	\label{ratio}
\end{table}

\textbf{Performance Comparison with Other Algorithms.} Now, we compare the performance of Algorithm~\ref{Alg1} with LP-II-GB, LP-OV-GB, and Varys. We set all the weights of coflows to be equal to one.

\textit{1. Performance evaluation under synthetic traffic:} For each of the two instances explained in Section~\ref{wldescription}, we randomly generate $100$ different traffic traces and compute the average performance of algorithms over the traffic traces.

Figure~\ref{synthwithoutrd} and~\ref{synthwithrd} depict the average result of our simulations (over $100$ dense and $100$ combined instances) for the zero release dates and general release dates, respectively. As we see, Algorithm~\ref{Alg1} (LP-OV-LS) outperforms Varys and LP-II-GB by almost $30\%$, and LP-OV-GB by almost $11 \%$ in dense instance for both general and zero release dates. In combined instance, the improvements are $35 \%$, $30\%$, and $17\%$ when all coflows are released at time $0$, and $28 \%$, $29\%$, and $17\%$ for the case of general release dates over Varys, LP-II-GB, and LP-OV-GB, respectively.
\begin{figure}[t]
	\centering
	\includegraphics[width=2.25in,height=1.45in]{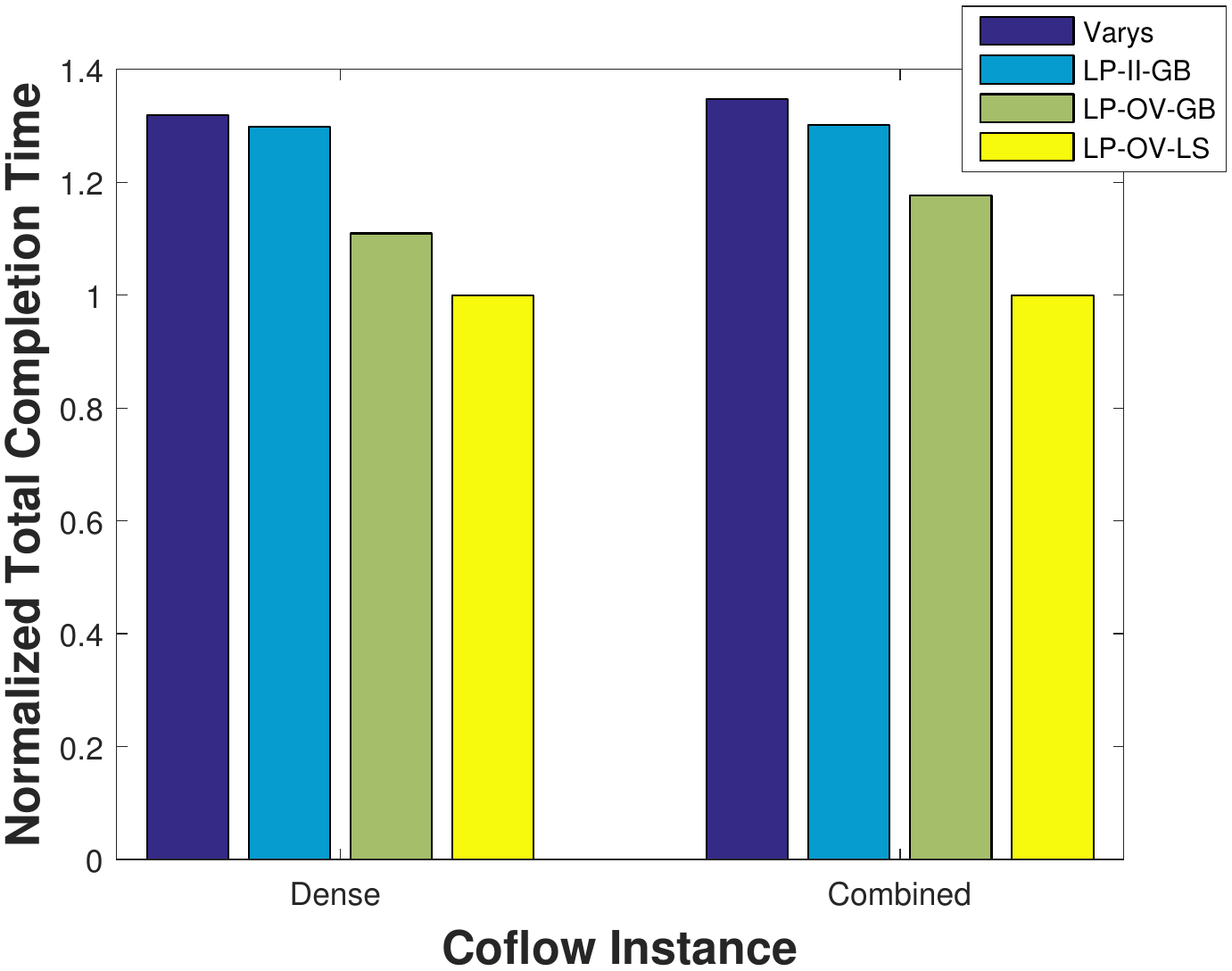}%
	\caption{Performance of Varys, LP-II-GB, LP-OV-GB, and LP-OV-LS when all coflows release at time $0$ for $100$ random dense and combined instances, normalized with the performance of LP-OV-LS.}
	\label{synthwithoutrd}%
\end{figure}
\begin{figure}[t]
	\centering
	\includegraphics[width=2.25in,height=1.45in]{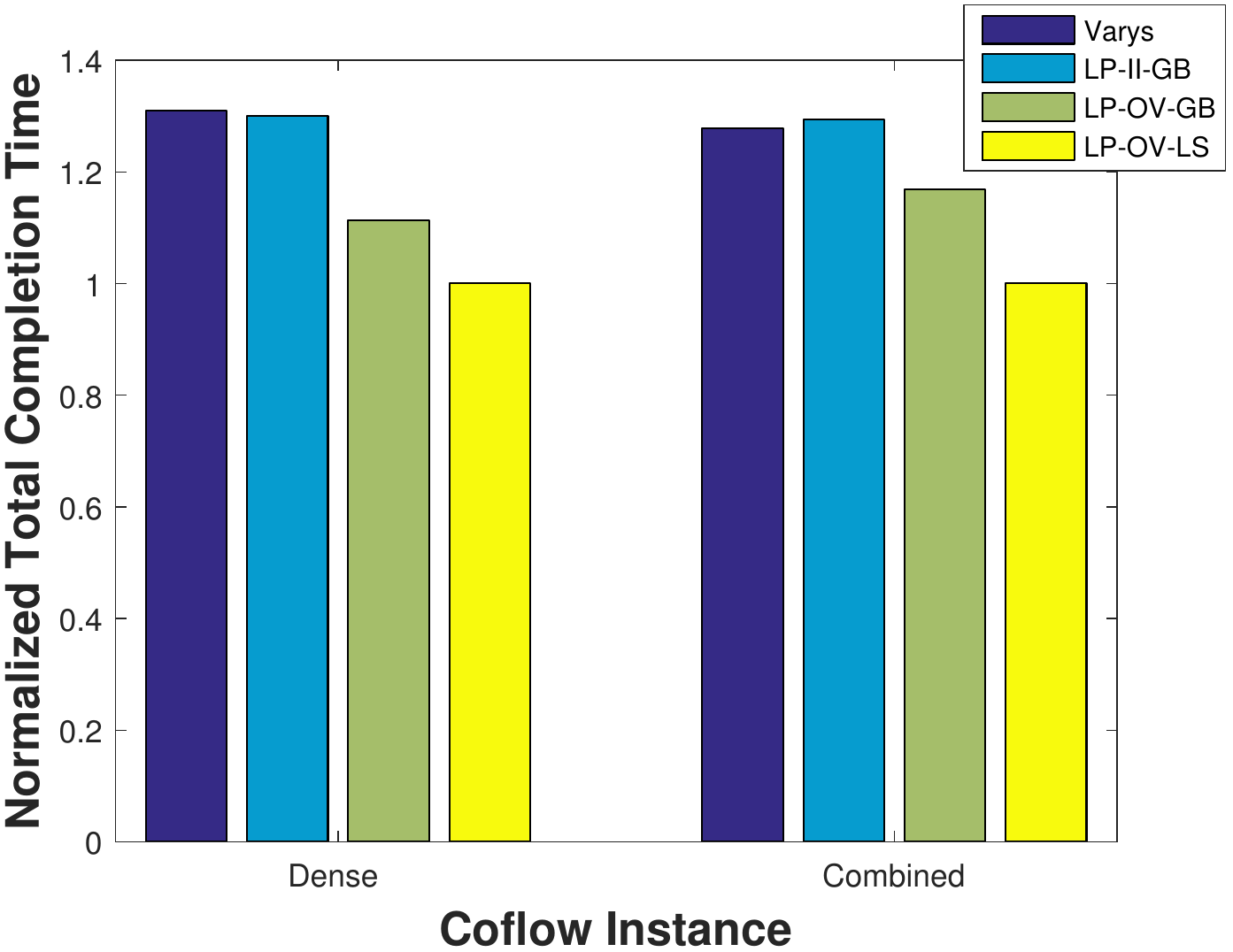}%
	\caption{Performance of Varys, LP-II-GB, LP-OV-GB, and LP-OV-LS in the case of release dates for $100$ random dense and combined instances, normalized with the performance of LP-OV-LS.}
	\label{synthwithrd}%
\end{figure}

This workload is more intensive in the number of non-zero flows; however, more uniform in the flow sizes and source-destination pairs in comparison to the real traffic trace. The real traffic trace (described in Section~\ref{wldescription}) contains a large number of sparse coflows; namely, about $50 \%$ of coflows have less than $10$ flows. Also, it widely varies in terms of flow sizes and source-destination pairs in the network. We now present evaluation results under this traffic.

\textit{2. Performance evaluation under real traffic:} We ran simulations for the four collections of coflows described in Section~\ref{wldescription}. We normalize the total completion time under each algorithm by the total completion time under Algorithm~\ref{Alg1} (LP-OV-LS).

Figure~\ref{withoutrd} shows the performance of different algorithms for different collections of coflows when all coflows are released at time $0$. LP-OV-LS outperforms Varys by almost $112-117 \%$ in different collections. It also constantly outperforms LP-II-GB and LP-OV-GB by almost $74-78 \%$ and $63-68\%$, respectively.

Figure~\ref{withrd} shows the performance of different algorithms for different collections of coflows for the case of release dates. LP-OV-LS outperforms Varys by almost $24 \%$, $65 \%$, $91 \%$, and $99 \%$ for all coflows, $M \geq 10$, $M \geq 30$, $M \geq 50$, respectively. It also outperforms LP-II-GB for $40 \%$, $62 \%$, $71 \%$, and $82 \%$, and LP-OV-GB by $19 \%$, $54 \%$, $64 \%$, and $73 \%$, respectively.
\begin{figure}[t]
	\centering
	\includegraphics[width=2.25 in,height=1.45in]{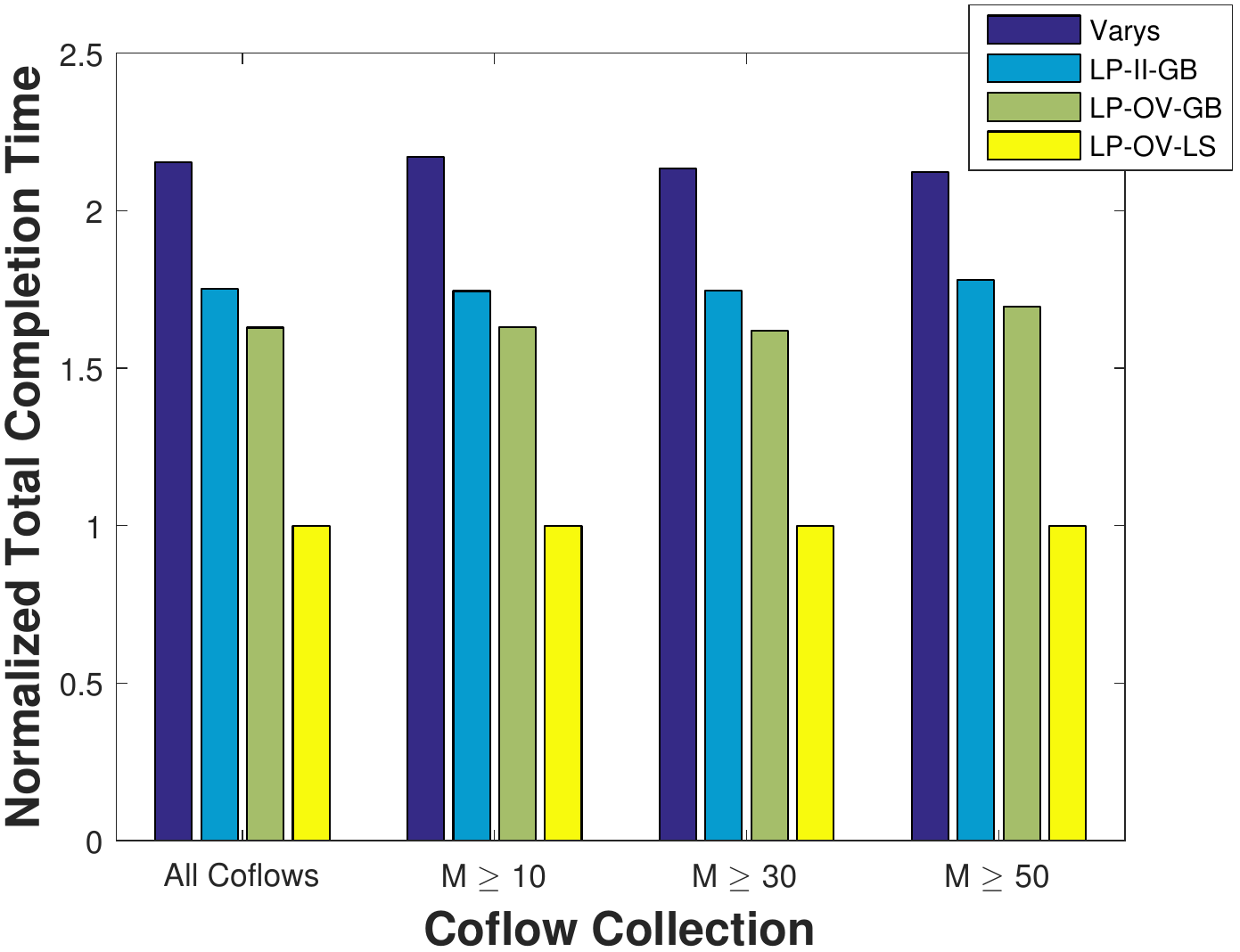}%
	\caption{Performance of Varys, LP-II-GB, LP-OV-GB, and LP-OV-LS when all coflows release at time $0$, normalized with the performance of LP-OV-LS, under real traffic trace.}
	\label{withoutrd}%
\end{figure}
\begin{figure}[t]
	\centering
	\includegraphics[width=2.25 in,height=1.45in]{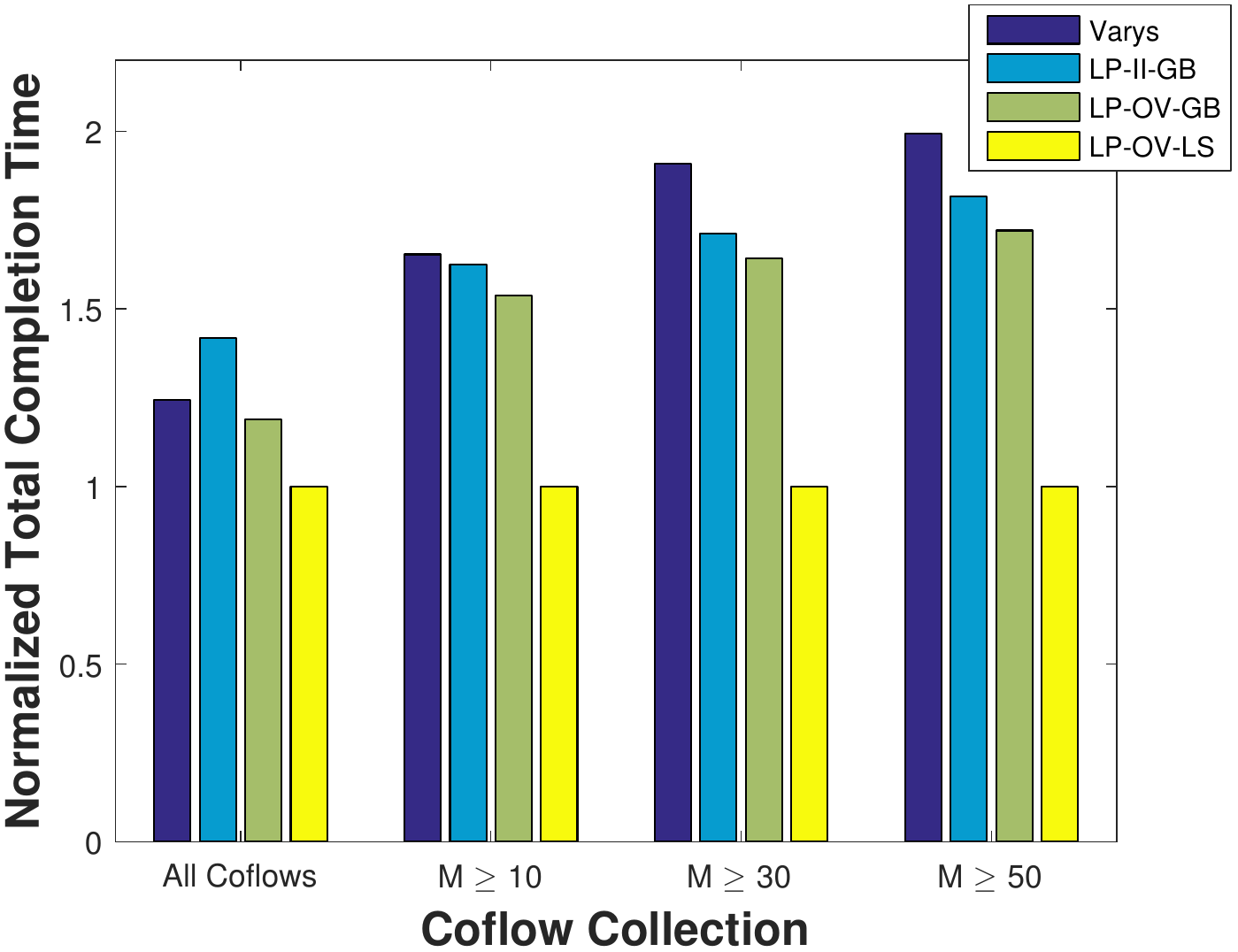}%
	\caption{Performance of Varys, LP-II-GB, LP-OV-GB, and LP-OV-LS in the case of release dates, normalized with the performance of LP-OV-LS, under real traffic trace.}
	\label{withrd}%
\end{figure} 
\section{Concluding Remarks}\label{conclusion}
In this paper, we studied the problem of scheduling of coflows with release dates to minimize their total weighted completion time, and proposed an algorithm with improved approximation ratio. We also conducted extensive experiments to evaluate the performance of our algorithm, compared with three algorithms proposed before, using both real and synthetic traffic traces. Our experimental results show that our algorithm in fact performs very close to optimal.

As future work, other realistic constraints such as precedence requirement or deadline constraints need to be considered. Also, theoretical and experimental evaluation of the performance of the proposed online algorithm is left for future work. While we modeled the datacenter network as a giant non-blocking switch (thus focusing on rate allocation), the routing of coflows in the datacenter network is also of great importance for achieving the quality of service.

\bibliography{IEEEabrv,Bibl}
\bibliographystyle{IEEEtran}

\appendix[NP-Completeness and Counter Example]
\label{2appexp}
\label{counterexample}

We first show NP-completeness of the coflow scheduling problem as formulated in optimization~\dref{optimization}. This is done through reduction from the concurrent open shop problem, in which a set of $K$ jobs and a set of $N$ machines are given. Each job consists of some tasks where each task is associated with a size and a specific machine in which it should be processed. We convert each job to a coflow by constructing a diagonal demand matrix~\cite{qiu2015minimizing}. By this construction, the constraints~\dref{capcons1} and~\dref{capcons2} are equivalent. The optimal solution to optimization~\dref{optimization} consists of none-negative transmission rates ${x^k_{j,j}}^\star(t)$ that sum to at most one on destination node $j$ at each time $t \in [0,T]$. However, in the concurrent open shop problem each machine can work on one task at a time which can be translated to zero and one transmission rates in the jargon of the coflow scheduling problem. Now, we show that given an optimal solution with rates ${x^k_{j,j}}^\star(t)$ to optimization~\dref{optimization} for the converted coflow scheduling problem, we can always transform it to a feasible solution for the original concurrent open shop problem. To do so, we consider destination node $j$ (machine $j$) and start from the last flow (task) that completes on this node. If there are multiple last flows, we choose one arbitrarily. We denote by ${f^k_{j,j}}^\star$ its optimal finishing time and by $d^k_{j,j}$ its size. We then set all transmission (processing) rates of this flow (task) to zero from time $0$ to ${f^k_{j,j}}^\star-d^k_{j,j}$, and to one from time ${f^k_{j,j}}^\star-d^k_{j,j}$ to ${f^k_{j,j}}^\star$. We adjust rates of other flows such that transmission rates sum to at most one at every time while all the flows are guaranteed to be processed before their completion time (which is given by the optimal solution). This can be easily done by increasing ${x^{k^\prime}_{j,j}}^\star(t)$ for $t \in [0,{f^k_{j,j}}^\star-d^k_{j,j}]$ by $\Delta x^{k^\prime}_{j,j}(t)$ determined as follows
\begin{equation*}
\Delta x^{k^\prime}_{j,j}(t)=\frac{\int_{{f^k_{j,j}}^\star-d^k_{j,j}}^{{f^k_{j,j}}^\star}{x^{k^\prime}_{j,j}}^\star(\tau) d\tau}{d^k_{j,j}} \times {x^{k}_{j,j}}^\star(t)
\end{equation*} 
By doing so, finishing time of the last flow does not change, and finishing time of other flows may decrease. The iterative procedure is repeated until processing rates of all flows converted to zero or one on node $j$. Therefore, we end up with possibly better solution in terms of total completion times of flows for node $j$ with zero-one rates. We apply this mechanism to all nodes; hence, the total completion time of the transformed solution is as good as the optimal solution. Thus, if an algorithm can solve the coflow scheduling problem in polynomial time, it can do so for concurrent open shop problem which contradicts with its NP-completeness. This completes the argument and NP-completeness of coflow scheduling problem is concluded.\\
Moreover, as we discussed in Section~\ref{Overview}, the $2$-approximation algorithms for the concurrent open shop problem cannot be directly applied to achieve $2$-approximation algorithms for the coflow scheduling problem. This is because given an ordering of $K$ coflows, there does not always exist a schedule in which the first coflow completes at time $W(1)$, the second coflow completes at time $W(1,2)$, and so on, until the last coflow completes at time $W(1,\cdots,K)$ (recall Definition~\ref{def2} for definition of $W(1,\cdots,k)$). We provide a counter example to show this.
\begin{example}
	Consider a $3 \times 3$ network with $2$ coflows as shown in Figure~\ref{network3}. One can force the ordering algorithm to output orange coflow as the first coflow and the green coflow as the second one in the list (e.g., by means of assigning appropriate weight to coflows). To finish the first coflow (orange coflow) in $W(1)$, transmission rates are assigned as shown in Figure~\ref{schedule1}. To avoid under utilization of network resources, the remaining capacities are dedicated to flows of coflow $2$ (green coflow). After $W(1)=2$ units of time, coflow $1$ completes and the remaining flows of coflow $2$ is as shown in Figure~\ref{schedule2}, therefore, one needs $2$ more units of time to complete remaining flows of coflow $2$. Hence, coflow $2$ completes at time $4 > W(1,2)=3$.
\begin{figure}[t]
	\centering
	\includegraphics[trim={0 1.1in 0 1.5in}, clip,width=2.6 in, height=1.3 in]{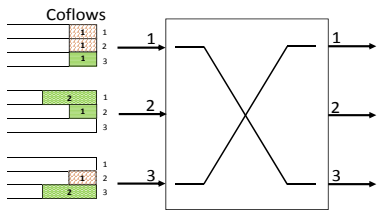}
	\caption{Two coflows in a $3 \times 3$ switch architecture. Flow sizes are depicted inside each flow.}
	\label{network3}
\end{figure}
\begin{figure}[t]
	\centering
	\begin{subfigure}[t]{1\columnwidth}
		\centering
		\includegraphics[trim={0 1.1in 0 1.5in}, clip,width=2.6 in, height=1.3 in]{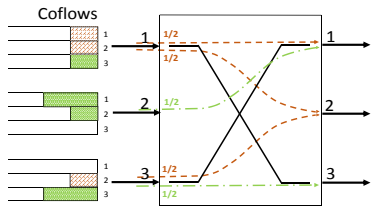}%
		\caption{Transmission rates so as to complete orange coflow at time $2$.}
		\label{schedule1}%
	\end{subfigure}\hfill
	% 	\vspace{-.1in}
	\begin{subfigure}[t]{1\columnwidth}
		\centering
		\includegraphics[trim={0 1.1in 0 1.5in}, clip,width=2.6 in, height=1.3 in]{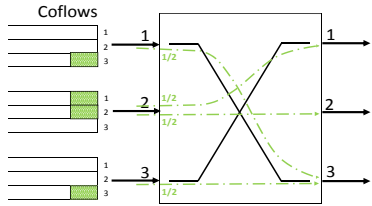}%
		\caption{Remaining flows of green coflow at time $2$ and rate assignment to complete its flows at time $4$.}%
		\label{schedule2}%
	\end{subfigure}
	\caption{Inaccuracy of proposed algorithm in~\cite{luo2016towards} .}
	\label{schedule}
	\vspace{-.2in}
\end{figure}
\end{example}
In fact, the $2$-approximation algorithm in~\cite{luo2016towards}, for coflow scheduling when all the release times are zero, relies on the assumption that such a schedule exists which, as we showed by the counter example, is not always true and hence the $4$-approximation algorithm proposed in this paper is the best known approximation algorithm in this case.

\end{document}